\documentclass[11pt]{article}
\usepackage{amsmath, amsthm, amssymb}
\usepackage{graphicx}
\usepackage{algorithm}
\usepackage{algpseudocode}
\usepackage[margin=1in]{geometry}
\usepackage{enumitem}
\usepackage{tikz}
\usepackage{pgfplots}
\pgfplotsset{compat=newest} 
\usetikzlibrary{pgfplots.groupplots}
\makeatletter
\pgfplotsset{
    groupplot xlabel/.initial={},
    every groupplot x label/.style={
        at={($({\pgfplots@group@name\space c1r\pgfplots@group@rows.west}|-{\pgfplots@group@name\space c1r\pgfplots@group@rows.outer south})!0.5!({\pgfplots@group@name\space c\pgfplots@group@columns r\pgfplots@group@rows.east}|-{\pgfplots@group@name\space c\pgfplots@group@columns r\pgfplots@group@rows.outer south})$)},
        anchor=north,
    },
    groupplot ylabel/.initial={},
    every groupplot y label/.style={
            rotate=90,
        at={($({\pgfplots@group@name\space c1r1.north}-|{\pgfplots@group@name\space c1r1.outer
west})!0.5!({\pgfplots@group@name\space c1r\pgfplots@group@rows.south}-|{\pgfplots@group@name\space c1r\pgfplots@group@rows.outer west})$)},
        anchor=south
    },
    execute at end groupplot/.code={%
      \node [/pgfplots/every groupplot x label]
{\pgfkeysvalueof{/pgfplots/groupplot xlabel}};  
      \node [/pgfplots/every groupplot y label] 
{\pgfkeysvalueof{/pgfplots/groupplot ylabel}};  
    }
}

\def\endpgfplots@environment@groupplot{%
    \endpgfplots@environment@opt%
    \pgfkeys{/pgfplots/execute at end groupplot}%
    \endgroup%
}
\makeatother

\usepackage{hyperref}

\theoremstyle{plain}
\newtheorem{theorem}{Theorem}
\newtheorem{lemma}[theorem]{Lemma}
\newtheorem{proposition}[theorem]{Proposition}
\newtheorem{corollary}[theorem]{Corollary}

\theoremstyle{definition}

\theoremstyle{remark}

\newcommand{\R}{\mathbb{R}}
\newcommand{\Z}{\mathbb{Z}}
\newcommand{\cD}{\mathcal{D}}
\newcommand{\cE}{\mathcal{E}}
\newcommand{\cF}{\mathcal{F}}
\newcommand{\cG}{\mathcal{G}}
\newcommand{\cJ}{\mathcal{J}}
\newcommand{\wtilde}[1]{\widetilde{#1}}
\newcommand{\eps}{\epsilon}
\newcommand{\norm}[1]{\| #1 \|}
\newcommand{\Norm}[1]{\left\| #1 \right\|}
\newcommand{\abs}[1]{| #1 |}
\newcommand{\Abs}[1]{\left| #1 \right|}
\DeclareMathOperator{\poly}{poly}
\DeclareMathOperator{\rk}{rank}

\DeclareMathOperator{\spn}{span}
\DeclareMathOperator*{\E}{\mathbb{E}}

\allowdisplaybreaks

\title{Single Pass Entrywise-Transformed Low Rank Approximation}

\author{
Yifei Jiang\thanks{Tandon School of Engineering, New York University. \texttt{jw7420@nyu.edu}}
\and
Yi Li\thanks{School of Physical and Mathematical Sciences, Nanyang Technological University. \texttt{yili@ntu.edu.sg}}
\and
Yiming Sun\thanks{School of Physical and Mathematical Sciences, Nanyang Technological University. \texttt{yiming005@e.ntu.edu.sg}}
\and
Jiaxin Wang \thanks{\texttt{w1059506278@gmail.com}}
\and
David P.\ Woodruff\thanks{Department of Computer Science, Carnegie Mellon University. \texttt{dwoodruf@andrew.cmu.edu}}
}

\date{}

\begin{document}

\maketitle

\begin{abstract}
	In applications such as natural language processing or computer vision, one is given a large $n \times d$ matrix $A = (a_{i,j})$ and would like to compute a matrix decomposition, e.g., a low rank approximation, of a function $f(A) = (f(a_{i,j}))$ applied entrywise to $A$. A very important special case is the likelihood function $f\left( A \right ) = \log{\left( \left| a_{ij}\right| +1\right)}$. A natural way to do this would be to simply apply $f$ to each entry of $A$, and then compute the matrix decomposition, but this requires storing all of $A$ as well as multiple passes over its entries. 
	Recent work of Liang et al.\ shows how to find a rank-$k$ factorization to $f(A)$ for an $n \times n$ matrix $A$ using only $n \cdot \poly(\epsilon^{-1}k\log n)$ words of memory, with overall error $10\|f(A)-[f(A)]_k\|_F^2 + \poly(\epsilon/k) \|f(A)\|_{1,2}^2$, where $[f(A)]_k$ is the best rank-$k$ approximation to $f(A)$ and $\|f(A)\|_{1,2}^2$ is the square of the sum of Euclidean lengths of rows of $f(A)$. Their algorithm uses three passes over the entries of $A$. The authors pose the open question of obtaining
	an algorithm with $n \cdot \poly(\eps^{-1}k\log n)$ words of memory using only a single pass over the entries of $A$.
	In this paper we resolve this open question, obtaining the first single-pass algorithm for this problem and
	for the same class of functions $f$ studied by Liang et al. Moreover, our error is $\|f(A)-[f(A)]_k\|_F^2 + \poly(\epsilon/k) \|f(A)\|_F^2$, where $\|f(A)\|_F^2$ is the sum of squares of Euclidean lengths of rows of $f(A)$. Thus our error is significantly smaller, as it removes the factor of $10$ and also $\|f(A)\|_F^2 \leq \|f(A)\|_{1,2}^2$. We also give an algorithm for regression, pointing out an error in previous work, and empirically validate our results. 
\end{abstract}

\section{Introduction}

There are numerous applications with matrices that are too large to fit in main memory, such as matrices arising in machine learning, image clustering, natural language processing, network analysis, and recommendation systems. This makes it difficult to process such matrices, and one common way of doing so is to stream through the entries of the  matrix one at a time and maintain a short summary or {\em sketch} that allows for further processing. In some of these applications, such as network analysis or distributed recommendation systems, matters are further complicated because one would like to take the difference or sum of two or more matrices over time. 

A common goal in such applications is to compute a low rank approximation to a large matrix $A \in \mathbb{R}^{n \times d}$. If the rank of the low rank approximation is $k$, then one can approximate $A$ as $U \cdot V$, where $U \in \mathbb{R}^{n \times k}$ and $V \in \mathbb{R}^{k \times d}$. This results in a parameter reduction, as $U$ and $V$ only have $(n+d)k$ parameters in total, as compared to the $nd$ parameters required of $A$. Since $k \ll \min(n,d)$, this parameter reduction is significant. Not only does it result in much smaller storage, when multiplying $A$ by a vector $x$, it also now only takes $O((n+d)k)$ time instead of $O(nd)$ time, since one can first compute $V\cdot x$ and then $U \cdot (V \cdot x)$. 

A challenge in the above applications is that often wants to compute a low rank approximation not to $A$, but to 
an {\em entrywise} transformation to $A$ by a function $f$. Namely, if $A = (a_{i,j})$, then we define $f(A) = f(a_{i,j})$ where we apply the function $f$ to each entry of $A$. Common functions $f$ include $f(x) = \log_2^c(|x| + 1)$ or $f(x) = |x|^{\alpha}$ for $0 \leq \alpha \leq 2$. Indeed, for word embeddings in natural language processing (NLP), an essential subroutine is to perform a low rank approximation of a matrix after applying the log-likelihood function to each entry, which corresponds to $f(x) = \log_2(|x|+1).$ Note that in NLP the input matrices
are often word co-occurrence count matrices, which can be created e.g., from the entire Wikipedia database. Thus, such
matrices are huge, with millions of rows and columns, and hard to store in memory. This necessitates models such as
the streaming model for processing such data. 

We can indeed capture the above scenarios formally with the {\em streaming model} of computation. In this model, there is
a large underlying matrix $A$, and we see a long sequence of updates to its coordinates in the form  $\left(i,j,\delta\right)$ with $\delta$ $\in$ $\left\{\pm1\right\}$, and 
representing the update $A_{i,j} \gets A_{i,j} + \delta$. Each pass over the data stream is very expensive, and 
thus one would like to minimize the number of such passes. Also, one would like to use as little memory as possible
to compute a low rank approximation of the transformed matrix $f(A)$ in this model. In this paper we will consider
approximately optimal low rank approximations, meaning factorizations $U \cdot V$ for which $\|U \cdot V - f(A)\|_F^2 \leq \|[f(A)]_k - f(A)\|_F^2 + \poly(\epsilon/k) \|f(A)\|_F^2,$ where $[f(A)]_k$ is the optimal rank-$k$ matrix approximating
$f(A)$ in Frobenius norm. Recall the Frobenius norm $\|B\|_F$ of a matrix $B$ is defined to be $(\sum_{i,j} B_{i.j}^2)^{1/2}$, which is the entrywise Euclidean norm of $B$. 

Although there is a body of work
in the streaming model computing low rank approximations of matrices \cite{bwz16,cw09,DMMW12,upadhyay2014differentially,w14}, such methods no longer apply in 
our setting due to the non-linearity of the function $f$. Indeed, a number of existing methods are based on 
dimensionality reduction, or {\em sketching}, whereby one stores $S \cdot A$ for a random matrix $S$ with a small number of rows. If there were no entrywise transformation applied to $A$, then given an update $A_{i,j} \gets A_{i,j} + \delta,$
one could simply update the sketch $(S \cdot A)_j \gets (S \cdot A)_j + S_i \cdot \delta,$ where $(S \cdot A)_j$ denotes the $j$-th column of $S \cdot A$ and $S_i$ denotes the $i$-th column of $S$. However, given an entrywise
transformation $f$, which may be non-linear, 
and given that we may see multiple updates to an entry of $A$, e.g., in the difference of
two datasets, it is not clear how to maintain $S \cdot f(A)$ in a stream. 

A natural question is: {\em can we compute a low-rank approximation to $f(A)$ in the streaming model with a small number of passes, ideally one, and a small amount of memory, ideally $n \cdot \textrm{poly}(k/\eps)$ memory?}

Motivated by the applications above, this question was asked by Liang et al.~\cite{lswyy20}; see also earlier work
which studies entrywise low rank approximation in the distributed model \cite{wz16}. 
The work of Liang et al.~\cite{lswyy20}
studies the function $f(x) = \log_2(|x|+1)$ and 
gives a three-pass algorithm for $n \times n$ matrices $A$ 
achieving $n \cdot \poly(\eps^{-1}k\log n)$ memory and outputting factors $U$ and $V$ with 
the following error guarantee:
\[
    \Norm{U \cdot V - f(A)}_F^2 \leq 10 \Norm{[f(A)]_k - f(A)}_F^2 + \poly(\eps/k)\Norm{f(A)}_{1,2}^2,
\]
where for a matrix $B$, $\|B\|_{1,2}$ is the sum of the Euclidean lengths of its columns. We note that
this error guarantee is considerably weaker than what we would like, as there is a multiplicative factor
$10$ and an additive error that depends on $\|f(A)\|_{1,2}$. Using the relationship between the $1$-norm 
and the $2$-norm, we have that $\|f(A)\|_{1,2}$ could be as large as $\sqrt{n} \|f(A)\|_F$, and so their
additive error can be a $\sqrt{n}$ factor larger than what is desired. Also, although the memory is of the desired order, the fact that the algorithm requires $3$ passes can significantly slow it down. 
Moreover, when data is crawled from the
internet, e.g. in applications of network traffic, it may be impractical to store the entire data set \cite{DL09}. Therefore, in these settings it is impossible to make more than one pass over the data. Liang et al.~\cite{lswyy20} say ``Whether there exists a one-pass algorithm is still an open problem, and is left for future work.''

\subsection{Our Contributions}
In this paper, we resolve this main open question of \cite{lswyy20}, obtaining a {\em one}-pass algorithm achieving 
$(n+d) \cdot \poly(\eps^{-1}k\log n)$ memory for outputting a low rank approximation for the function 
$f(x) = \log_2(|x|+1)$, and achieving the stronger error guarantee:
\[
    \Norm{U \cdot V - f(A)}_F^2 \leq \Norm{[f(A)]_k - f(A)}_F^2 + \poly(\eps/k)\Norm{f(A)}_F^2.
\]
We note that the $\poly(\eps/k)$ factor in both the algorithm of \cite{lswyy20} and our algorithm can be made
arbitrarily small by increasing the memory by a $\poly(k/\eps)$ factor, and thus it suffices to
consider error of the form $\|[f(A)]_k - f(A)\|_F^2 + \epsilon \|f(A)\|_F^2$. We also note that our algorithm can be trivially adapted to rectangular matrices, so for ease of notation, we focus on the case $n = d$. 

At a conceptual level, the algorithm of \cite{lswyy20} uses one pass to obtain so-called 
approximate leverage scores of $f(A)$, then a second pass to sample columns of $f(A)$ according to these, and finally a third pass to do so-called adaptive sampling. In contrast, we observe that one can just do squared column
norm sampling of $f(A)$ to obtain the above error guarantee, which is a common method for low rank approximation
to $A$. However, in one pass it is not possible to sample
actual columns of $A$ or of $f(A)$ according to these probabilities, so we build a data structure to sample {\it noisy} columns by approximations to their squared norms in a single
pass. This is related to block $\ell_2$-sampling in a stream, see, e.g., \cite{mrwz20}. However, the situation here is complicated by the fact that we must sample according to the sum of squares of $f$ values of entries in a column of $A$, rather than the squared length of the column of $A$ itself. The transformation function $f$'s nonlinearity makes many of the techniques considered in \cite{mrwz20} inapplicable. To this end we build new hashing and sub-sampling data structures, generalizing data structures for length or squared length sampling from \cite{abiw09,a10,lsw18}, and we give a novel analysis for sampling noisy columns of $A$ proportional to the sum of squares of $f$ values to their entries. 

Additionally, we empirically validate our algorithm on a real-world data set, demonstrating significant advantages over the algorithm of Liang et al.~\cite{lswyy20} in practice. 

Finally, we apply our new sampling techniques to the regression problem, showing that our techniques are more broadly applicable. Although Liang et al.~\cite{lswyy20} claim a result for regression, we point out an error in their analysis\footnote{We have confirmed this with the authors.}, showing that their algorithm obtains larger error than claimed, and the error of our regression algorithm is considerably smaller.

\section{Preliminaries}
\paragraph{Notation.} We use $[n]$ to denote the set $\{1,2, \ldots, n\}$. For a vector $x\in \mathbb{R}^n$, we denote by $|x|$ the vector whose $i$-th entry is $|x_i|$. For a matrix $A\in \mathbb{R}^{N \times N}$, let $A_{i,*}$ be the $i$-th row of $A$ and $A_{*,j}$ be the $j$-th column of $A$. We also sometimes abbreviate $A_{i,\ast}$ or $A_{\ast,i}$ as $A_i$. We also define the norms $\|A\|_{p,q} = (\sum_i \|A_{i,*}\|_q^p)^{1/p}$, of which the Frobenius norm $\|A\|_{F} = \|A\|_{2,2}$ is a special case. Note that $\|A\|_{p,q} \equiv \|A^\top \|_{q,p}$. 
We use $\sigma_i(A)$ to denote the $i$-th singular value of $A$ and $[A]_k$ to denote the best rank-$k$ approximation to $A$. For a function $f$, let $f(A)$ denote the entrywise-transformed matrix $(f(A))_{ij} = f(A_{ij})$. 

\paragraph{Low-rank Approximation.} Given a matrix $M\in \R^{n\times n}$, an integer $k\leq n$ and an accuracy parameter $\epsilon > 0$, our goal is to output an $n \times k$ matrix $L$ with orthonormal columns for which $\|M - LL^\top M\|_{F}^{2} \leq \|M - [M]_k\|_{F}^{2} + \epsilon \|M\|_F^2$, 
where $[M]_k = \arg\,\min_{\rk(M^\prime)\leq k}\|M-M^\prime\|_{F}^{2}$. Thus, $LL^\top $ provides a rank-$k$ matrix to project the columns of $M$ onto. The rows of $L^\top M$ can be thought of as the ``latent features" in applications, and the rank-$k$ matrix $LL^\top M$ can be factored as $L \cdot (L^\top M)$, where $L^\top M$ is a $k \times n$ matrix and $L$ is an $n \times k$ matrix. 

\paragraph{Best Low-rank Approximation.} 
Consider the singular value decomposition $G=\sum_{t=1}^r \sigma_t u_t v_t^\top $, where $\sigma_1 \geq \sigma_2 \geq \cdots \geq \sigma_r > 0$ are the nonzero singular values of $G$, and $\{u_t\}$ and $\{v_t\}$ are orthonormal sets of column vectors such that $G^\top  u_t = \sigma_t v_t$ and $G v_t = \sigma_t u_t$ for each $t\leq r$. The vectors $\{u_t\}$ and $\{v_t\}$ are called the left and the right singular values of $G$, respectively.  
By the Eckart-Young theorem, for any rotationally invariant norm $\|\cdot\|$, the matrix $D_k$ that minimizes $\|G-D_k\|$ among all matrices $D$ of rank at most $k$ is given by $D_k=\sum_{t=1}^k G v_t v_t^\top $. This implies that 
$\|G-D_k\|_F^2=\sum_{t=k+1}^r\sigma_t^2$. 

For a matrix $A$, we denote by $\|A\|_2$ its operator norm, which is equal to its largest singular value. 

\paragraph{Useful Inequalities.} The first one is the matrix Bernstein inequality. 
\begin{lemma}[Matrix Bernstein]
\label{lem:matrix-bound}
Let $X_1$, $\dots$, $X_s$ be $s$ independent copies of a symmetric random matrix $X\in \mathbb{R}^{d\times d}$ with $\E (X)=\mathbf{0}$, $\Norm{X}_2 \leq \rho$ and $\Norm{\E (X^2)}_2 \leq \sigma^2$. Let $W=\frac{1}{s}\sum_{i=1}^sX_i$. For any $\eps>0$, it holds that
\[
    \Pr \{\Norm{W}_2>\eps \} \leq 2d\cdot e^{-s\eps^2/(\sigma^2+\rho\eps/3)}.
\]
\end{lemma}

Below we list a few useful inequalities regarding the function $f(x) = \log(1+|x|)$. The proofs are postponed to Appendix~\ref{sec:ineq proofs}.
\begin{proposition}\label{prop:log_ineq}
For $x > 0$, it holds that $\ln(1+x) > x/(1+2x)$.
\end{proposition}

\begin{proposition}\label{lem:log_additive}
It holds for all $x, y\in\R$ and all $a\geq 0$ that $f(x+y)\leq f(x)+f(y)$ and $f(ax)\leq a f(x)$. As a consequence, for $x,y\in \R^n$ it holds that $\norm{f(x+y)}_2^2\leq (\norm{f(x)}_2 + \norm{f(y)}_2)^2$.
\end{proposition}

\begin{proposition}\label{prop:log^2}
It holds for all $x,y\geq 0$ that $f(\sqrt{x^2+y^2})^2\leq f(x)^2 + f(y)^2$.
\end{proposition}

\begin{lemma}\label{lem:light_noise_f}
Let $a_1,\dots,a_m$ be real numbers and $\eps_1,\dots,\eps_m$ be $4$-wise independent random variables on the set $\{-1,1\}$ (i.e., Rademacher random variables). It holds that
\[
\E f\left(\sum_i \eps_i a_i\right)^2 \leq C \sum_i f(a_i)^2.
\]
where $C>0$ is an absolute constant. 
\end{lemma}

\begin{lemma}\label{lem:add_noise}
For arbitrary vectors $y, z \in \R^{n} $ such that
$\norm{f(y)}_2^2 \ge \xi^{-2} \norm{f(z)}_2^2$ for some $\xi\in(0,1)$, it holds that 
$(1-3\xi^{2/3})\norm{f(y)}_2^2 \leq \|f(y+z)\|_2^2 \leq (1+3\xi) \norm{f(y)}_2^2$.
\end{lemma}


\section{Algorithm}
Our algorithm uses two important subroutines: a subsampling data structure called an $H$-\textsf{Sketch}, and a sketch for approximating the inner product of a transformed vector and a raw vector called \textsf{LogSum}. The former is inspired from a subsampling algorithm of \cite{lsw18} and is meant to sample a noisy approximation to a column from a distribution which is close to the desired distribution. In fact, one can show that it is impossible to sample the actual columns in a single pass, hence, we have to resort to noisy approximations  and show they suffice. 
The latter \textsf{LogSum} sketch is the same as in~\cite{lswyy20}, which approximates the inner product $\langle f(x),y\rangle$ for vectors $x,y$.
Executing these sketches in  
parallel is highly non-trivial since the subsampling algorithm of \cite{lsw18} samples columns of $A$ according
to their $\ell_2$ norms, but here we must sample them according to the squares of their $\ell_2$ norms after
applying $f$ to each entry. 

Roughly speaking, the above combination gives us a small set of $\poly(k/\epsilon)$ noisy columns of $f(A)$, sampled approximately from the squared $\ell_2$ norm of each column of $f(A)$, after which we can appeal to squared column-norm sampling results for low rank approximation in \cite{fkv04}, which argue that if you then compute the top-$k$ left singular vectors of $f(A)$, forming the columns of the $n \times k$ matrix $L$, then $LL^\top f(A)$ is a good rank-$k$ approximation to $f(A)$. The final output of the low-rank approximation will be two factors, $L$ and $L^\top f(A)$. The algorithm in~\cite{lswyy20} first computes $L$ by an involved algorithm in three passes, and then computes $L^\top f(A)$ in another pass using \textsc{LogSum} sketches. Our algorithm follows the same outline but we shall demonstrate how to compute $L$ in only one pass, which is our sole focus for low-rank approximation in this paper. Note that our ultimate goal, which we only achieve approximately, is to sample columns of $f(A)$ proportional to their squared $\ell_2$ norms. This is a fundamentally different sampling scheme from that of \cite{lswyy20}, which performs leverage score sampling followed by adaptive sampling, which are not amenable to a one-pass implementation. 

\subsection{$H$-\textsf{Sketch}}\label{sec:h-sketch}

\begin{algorithm}[t]
	\caption{Basic heavy hitter substructure} \label{alg:h-sketch}
	\textbf{Input:} $A \in \mathbb{R}^{n\times n}$, $\nu$, $\phi$ \\
	\textbf{Output:} a data structure $H$
	\begin{algorithmic}[1]
		\State $w\gets O(1/(\phi^2 \nu^3))$
		\State Prepare a pairwise independent hash function $h: [n]\to [w]$
		\State Prepare $4$-wise independent random signs $\{\eps_i\}_{i=1}^n$
		\State Prepare a hash table $H$ with $w$ buckets, where each bucket stores a vector in $\R^n$.
		\For{each $v\in [w]$}
			\State $H_v\gets \sum_{i \in h^{-1}(v)} \eps_i A_{i}$
		\EndFor
		\State \Return $H$
	\end{algorithmic}
\end{algorithm}

\begin{algorithm}[t]
	\caption{Complete heavy hitter structure $\cD$} \label{alg:h-sketch-extended}
	\textbf{Input:} $A \in \mathbb{R}^{n\times n}$, $\nu$, $\phi$, $\delta$ \\
	\textbf{Output:} a data structure $H$
	\begin{algorithmic}[1]
		\State $R\gets O(\log(n/\delta))$
		\For{each $r \in [R]$}
			\State Initialize a basic substructure $H^{(r)}$ (Algorithm~\ref{alg:h-sketch}) with parameters $\nu$ and $\phi$
		\EndFor
		\Statex
		\Function{Query}{$i$}
			\For{each $r\in [R]$}
				\State $v_r\gets H^{(r)}_{h_r(i)}$ \Comment{$h_r$ is the hash function in $H^{(r)}$}
			\EndFor
			\State $r^\ast \gets \text{index $r$ of the median of $\{\norm{f(v_r)}_2\}_{r\in[R]}$}$ \label{alg:h-sketch-extended:median}
			\State \Return $v_{r^\ast}$
		\EndFunction
	\end{algorithmic}
\end{algorithm}

\begin{algorithm}[t]
	\caption{Sampling using $H$-\textsf{Sketch}}\label{alg:h-sketch-sample}
	\begin{algorithmic}[1]
		\Require (i) An estimate $\widehat{M}$ such that $M\leq \widehat{M}\leq KM$; (ii) a complete heavy hitter structure $\cD_0$ of parameters $(O(1), O(\eps^3/(KL^3)), 1/(\hat L+1))$; (iii) $\hat L$ complete heavy hitter structures (see Algorithm~\ref{alg:h-sketch}), denoted by $\cD_1,\dots,\cD_{\hat L}$, where $\cD_j$ $(j\in [\hat L])$ has parameters $(O(1),O(\eps^3/L^3),1/(\hat L+1))$ and is applied to the columns of $A$ downsampled at rate $2^{-j}$; 
		\State $L\gets \log(Kn/\eps)$, $\hat L\gets \log n$
		\State $\zeta\gets $ a random variable uniformly distributed in $[1/2,1]$
		\For{$j=0,\dots,\hat L$}
			\State $\Lambda_j \gets \text{top }\Theta(L^3/\eps^3)\text{ heavy hitters from }\cD_j$
		\EndFor
		\State $j_0\gets \log(4K\eps^{-3}L^3)$
		\State $\zeta\gets$ uniform variable in $[1/2,1]$
		\For{$j = 1,\dots, j_0$}
			\State Let $\lambda_1^{(j)},\dots,\lambda_s^{(j)}$ be the elements in $\Lambda_0$ contained in $[(1+\eps)\zeta\frac{\widehat{M}}{2^j}, (2-\eps)\zeta\frac{\widehat{M}}{2^j}]$
			\State $\wtilde{M}_j \gets \abs{\lambda_1^{(j)}}+\cdots+\abs{\lambda_s^{(j)}}$
		\EndFor
		\For{$j = j_0 + 1,\dots, L$}
			\State Find the largest $\ell$ for which $\Lambda_\ell$ contains $s_j$ elements $\lambda^{(j)}_1,\dots,\lambda^{(j)}_{s_j}$ in $[(1+\eps)\zeta\frac{\widehat{M}}{2^j}, (2-\eps)\zeta\frac{\widehat{M}}{2^j}]$ for $(1-\sqrt{20}\eps)L^2/\eps^2\leq s_j \leq 2(1+\sqrt{20}\eps)L^2/\eps^2$
			\If{such an $\ell$ exists}
				\State $\wtilde{M}_j \gets (\abs{\lambda^{(j)}_1}+\cdots+\abs{\lambda^{(j)}_{s_j}})2^\ell$
				\State $W_j\gets \Lambda_\ell$
			\Else
				\State $\wtilde{M}_j \gets 0$
			\EndIf
		\EndFor
		\State $j^\ast \gets \text{sample from $[L]$ according to pdf } \Pr(j^\ast = j) = \wtilde{M}_j/\sum_j \wtilde{M}_j$
		\State $i^\ast \gets \text{sample from $W_{j^\ast}$ according to pdf } \Pr(i^\ast = i) = \abs{\lambda_i^{(j)}} / \wtilde{M}_j$
		\State $v_{j^\ast,i^\ast}\gets \text{vector returned by } \Call{Query}{i^\ast} \text{ on } \cD_{j^\ast}$
		\State \Return $v_{j^\ast,i^\ast}$
	\end{algorithmic}
\end{algorithm}

We first present a basic heavy hitter structure in Algorithm~\ref{alg:h-sketch}, and a complete heavy hitter structure in Algorithm~\ref{alg:h-sketch-extended} by repeating the basic structure $R$ times. The complete heavy hitter structure supports a query function. Below we analyze the guarantee of this heavy hitter data structure. 

Let $M = \norm{f(A)}_F^2$. We define $I_\eps = \{i\in [n]: \norm{f(A_i)}_2^2 \geq \eps M\}$, the set of the indices of the $\eps$-heavy columns. Let $\alpha$ be a small constant to be determined later.

\begin{lemma}\label{lem:heavy_no_collision}
With probability at least $0.9$, all columns in $I_{\alpha\phi}$ are isolated from each other under $h$.
\end{lemma}
\begin{proof}
Note that $|I_{\alpha\phi}| \leq 1/(\alpha\phi)$. Thus, there exists a collision with probability at most
\[
 \frac{1}{w}\binom{1/(\alpha\phi)}{2} \le \frac{1}{2w\alpha^2\phi^2}\le 0.1,
\]
provided that $w\geq 1/(0.2\cdot\alpha^2\phi^2) = 5/(\alpha^2 \phi^2)$.
\end{proof}

\begin{lemma}\label{lem:HH_single_repetition}
For each $u\in [n]$, it holds with probability at least $2/3$ that
\[
\Norm{f\left(\sum_{i\not\in (I_{\alpha\phi}\cup\{u\})} \mathbf{1}_{\{h(i) = h(u)\}} \eps_{i} A_i\right)}_2^2 \le 3C\frac{M}{w},
\]
where $C > 0$ is an absolute constant.
\end{lemma}
\begin{proof}
Let $v=h(u)$. Since $h$ is pairwise independent, $\Pr\{h(i)=v\}=1/w$ for all $i\neq w$. Let
\[
Z_v = \sum_{i\not\in (I_{\alpha\phi}\cup\{u\})}\mathbf{1}_{\{h(i) = v\}} \|f(A_i)\|_2^2.
\]
then
\[
\E Z_v \leq \sum_{i\not\in I_{\alpha\phi}} \E \mathbf{1}_{\{h(i) = v\}} \Norm{f(A_i)}_{2}^2 \leq \frac{M}{w}.
\]
It follows from Lemma~\ref{lem:light_noise_f} that
\begin{align*}
 \E_{\{\eps_i\}, h} \Norm{f\left(\sum_{i\not\in I_{\alpha\phi}} \mathbf{1}_{\{h(i) = v\}} \eps_i A_i \right)}_2^2  
&\leq \E_h C\sum_{i\not\in I_{\alpha\phi}} \Norm{f\left(\mathbf{1}_{\{h(i) = v\}} A_i\right)}_2^2 \\
&= C \E_h Z_v \\
&\leq C\frac{M}{w},
\end{align*}
where we used the fact that $f(0) = 0$ and $\mathbf{1}_{\{h(i) = v\}}\in\{0,1\}$ in the second step (the equality). The result follows from  Markov's inequality.
\end{proof}

\begin{lemma}\label{lem:HH-guarantee}
Suppose that $\nu\in (0, 0.05]$ and $\alpha = 0.3/C > \beta$, where $C$ is the absolute constant in Lemma~\ref{lem:HH_single_repetition}. With probability at least $1-\delta$, for all $i\in [n]$, the output $v_{r^\ast}$ of Algorithm~\ref{alg:h-sketch-extended} satisfies that
\begin{enumerate}[label=(\alph*)]
	\item $(1-\nu)\norm{f(A_i)}_2^2 \leq \norm{f(v_{r^\ast})}_2^2 \leq (1+\nu)\norm{f(A_i)}_2^2$ for all $i\in I_\phi$;
	\item $\norm{f(v_{r^\ast})}_2^2 \leq 0.92\phi M$ for all $i\not\in I_{\alpha\phi}$;
	\item $\norm{f(v_{r^\ast})}_2^2 \leq (1+\nu^{3/2})^2\phi M$.
\end{enumerate}
\end{lemma}
\begin{proof}
Fix $u\in [n]$. With probability at least $0.9-1/3 > 0.5$, the events in Lemmata~\ref{lem:heavy_no_collision} and~\ref{lem:HH_single_repetition} happen. Condition on those events.

From the proof of Lemma~\ref{lem:HH_single_repetition}, we know that $i\in I_\phi$ do not collide with other elements in $I_{\alpha\phi}$. Hence, it follows from Lemma~\ref{lem:add_noise} (where $\xi^2 \leq 3C/(\phi w)\leq (\nu/3)^{3}$) that
\[
(1-\nu)\Norm{f(A_i)}_2^2 \leq \Norm{f(H_{h(i)})}_2^2 \leq (1+\nu)\Norm{f(A_i)}_2^2,
\]
provided that $w\geq 3^4 C/(\phi \nu^{3})$.

When $i\not\in I_{\alpha\phi}$, we have 
\[
    \Norm{f(H_{h(i)})}_2^2 \leq 3C\left(\alpha\phi + \frac{1}{w}\right) M \leq (0.9 + \nu^{3/2})\phi M \leq 0.92\phi M.
\]

When $i\in I_{\alpha\phi}\setminus I_\phi$, we have that $H_i$ contains only $i$ and columns from $[n]\setminus I_{\alpha\phi}$. Hence by Proposition~\ref{lem:log_additive},
\[
    \Norm{f(H_{h(i)})}_2^2 \leq \left(\Norm{f(A_i)}_2 + \sqrt{\frac{3C}{w} M} \right)^2 
\leq \left(\sqrt{\phi M} + \sqrt{\nu^3\phi M} \right)^2 \\
\leq (1+\nu^{3/2})^2\phi M,
\]
provided that $w\geq 3C/(\phi \nu^3)$.

Finally, repeating $O(\log(n/\delta))$ times and taking the median and a union bound over all $n$ columns gives the claimed result.
\end{proof}

Next we analyze the sampling algorithm, presented in Algorithm~\ref{alg:h-sketch-sample}, which simulates sampling a column from $A$ according to the column norms. The following theorem is our guarantee.

\begin{theorem}\label{thm:main-sampling}
Let $\eps>0$ be a constant small enough. Algorithm~\ref{alg:h-sketch-sample} outputs $v_{j^\ast,i^\ast}$ which satisfies that, with probability at least $0.9$, there exists $u\in [n]$ such that 
\[
(1-O(\eps))\Norm{f(A_u)}_2^2 \leq \Norm{f(v_{j^\ast,i^\ast})}_2^2 
\leq (1+O(\eps))\Norm{f(A_u)}_2^2.
\]
Furthermore, there exists an absolute constant $c \in (0,1/2]$ such that
\[
\Pr\{u = i\} \geq c\frac{\Norm{f(A_i)}_2^2}{\Norm{f(A)}_F^2}
\]
for all $i$ belonging to some set $I\subseteq [n]$ such that $\sum_{i\in I} \norm{f(A_i)}_2^2\geq (1- 6\eps)M$, provided that $\eps$ further satisfies that $\eps \leq c/C$ for some absolute constant $C > 0$.
\end{theorem}
\begin{proof}
The analysis of the algorithm is largely classical, for which we define the following notions:
\begin{enumerate}[label=(\arabic*)]
	\item $T_{j} = \zeta M/2^j$;
	\item $S_{j} = \left \{i \in [n] : \|f(A_{i})\|_2^2 \in \left ( T_{j},2T_{j} \right ] \right \} $ is the $j$-th {\it level set} of $A$;
	\item a level $j\in [L]$ is \emph{important} if $|S_j| \geq \eps 2^j/L$;
	\item $\cJ\subseteq [L]$ is the set of all important levels.
\end{enumerate}

It follows from the argument in~\cite{LWY21}, or an argument similar to~\cite{abiw09} that the columns we miss contribute to only an $O(\eps)$-fraction of the norm, and for each level $j \in \{1,\dots,j_0\}\cup \cJ$, each of the recovered columns $\lambda_i$ ($i\in [s_j]$) corresponds to some $u = u(i)\in S_j$ and satisfies that $(1-O(\eps))\norm{f(A_{u_i})}_2^2 \leq \lambda_i \leq (1+O(\eps))\norm{f(A_{u_i})}_2^2$. 

Next we prove the second part. For a fixed $i\in [n]$, define events 
\[
\cE_i = \{\text{$i$ falls in a level $j\in [j_0]\cup \cJ$}\}
\]
and a set of ``good'' columns
\[
I = \{i: \Pr\{\cE_i\} \geq \beta\}
\]
for some constant $\beta\leq 1/2$. Since all non-important levels always contribute to at most a $2\eps$-fraction of $M$, it follows that the bad columns contribute to at most a $2\eps/(1-\beta)$-fraction of $M$, that is,
\[
\sum_{i\not\in I} \Norm{f(A_i)}_2^2 \leq \frac{2\eps}{1-\beta}\cdot M.
\]

Next we define the event that
\[
\cF_i = \{\cE_i \text{ and }\norm{f(A_i)}_2^2 \in [(1+\eps)T_j,2(1-\eps)T_j]\},
\]
then it holds for all $i\in I$ that $\Pr\{\cF_i\}\geq \Pr\{\cE_i\} - O(\eps) \geq 0.9\beta$ for $\eps$ sufficiently small.

Let $\cG_j$ denote the event that the magnitude level $j$ is chosen, and $j(i)$ is the index of the magnitude level containing column $i$. Then for those $i$'s with $j = j(i)\in [j_0]\cup \cJ$,
\[
\Pr\{\cG_{j}|\cF_i\} = \frac{(1\pm O(\eps))\wtilde{M}_{j}}{(1\pm O(\eps))\sum_j \wtilde{M}_j} = \frac{(1\pm O(\eps))M_{j}}{(1\pm O(\eps))M} 
= (1\pm O(\eps))\frac{M_{j}}{M}
\]
and
\[
\Pr\{u = i \mid \cG_{j}\cap \cF_i\} = \frac{\lambda^{(j)}_t}{\wtilde{M}_j} = \frac{1\pm O(\eps)\norm{f(A_i)}_2^2}{(1\pm O(\eps))M_{j}} = (1\pm O(\eps))\frac{\norm{f(A_i)}_2^2}{M_j}
\]
Hence
\[
\Pr\{u = i\} = \Pr\{u = i \mid \cG_{j} \cap \cF_i\}\Pr\{\cG_{j}|\cF_i\}\Pr\{\cF_i\} 
\geq 0.9\beta (1 - O(\eps)) \frac{\norm{f(A_i)}_2^2}{M}\geq 0.8\beta\frac{\norm{f(A_i)}_2^2}{M},
\]
provided that $\eps$ is sufficiently small.
\end{proof}

Now we show how to obtain an overestimate $\widehat{M}$ for $M$. We assume that all entries of $A$ are integer multiples of $\eta = 1/\poly(n)$ and are bounded by $\poly(n)$, which is a common and necessary assumption for streaming algorithms, otherwise storing a single number would take too much space. Let $\tilde{f}(x) = \log^2(1+|\eta x|)$, then $\norm{f(A)}_F^2 = \sum_{i,j} \tilde{f}(\eta^{-1}A)$, where $\eta^{-1}A$ has integer entries. Hence, we can run the algorithm implied by Theorem 2 of~\cite{BCWY16} on $\eta^{-1}A$ in parallel in order to obtain a constant-factor estimate to $\norm{f(A)}_F^2$. To justify this application of the theorem, we verify in Appendix~\ref{sec:1-pass-tractable} that the function $\tilde{f}(|x|)$ is slow-jumping, slow-dropping and predictable on nonnegative integers as defined by~\cite{BCWY16}.

Finally, we calculate the sketch length. The overall sketch length is dominated by that of Algorithm~\ref{alg:h-sketch-sample}. In Algorithm~\ref{alg:h-sketch-sample}, there are $\hat L = O(\eps^{-1}\log n)$ heavy hitter structures $\cD_1,\dots,\cD_{\hat L}$, each of which has a sketch length of $O(1/(\phi^2\nu^3) \log(nL)) = \poly(L,1/\eps,\log n) = \poly(\log n,1/\eps)$. There is an additional heavy hitter structure $\cD_0$ of sketch length $O(\poly(K, L, 1/\eps,\log n)) = \poly(\log n,1/\eps)$. Hence the overall sketch length is $\poly(\log n,1/\eps)$. Each cell of the sketch stores an $n$-dimensional vector. We summarize this in the following theorem.

\begin{theorem}\label{thm:h sketch}
Suppose that $A\in (\eta\Z)^{n\times n} $ with $|A_{ij}|\leq \poly(n)$ is given in a turnstile stream, where $\eta = 1/\poly(n)$. There exists a randomized sketching algorithm which maintains a sketch of $n\poly(\eps^{-1}\log n)$ space and outputs a vector $v_{j^\ast,i^\ast}\in \R^n$ which satisfies the same guarantee as given in Theorem~\ref{thm:main-sampling}.
\end{theorem}

\subsection{Low-Rank Approximation}\label{sec:low-rank}

Suppose that we have an approximate sampling of the rows of $f(A)$ so that we obtain a sample $f(A_i)+E_i$ with probability $p_i$ satisfying
\begin{equation}\label{eqn:sampling_distribution}
p_i \geq c\frac{\Norm{f(A_i)}_2^2}{\Norm{f(A)}^2_F}
\end{equation}
for some absolute constant $c\leq 1$. The $p_i$'s are known to us (if $c = 1$, then we do not need to know the $p_i$).

The following is our main theorem in this section, which is analogous to Theorem 2 of \cite{fkv04}. The proof is postponed to Appendix~\ref{sec:proof of low rank raw}.
\begin{theorem}\label{thm:low-rank-raw}
Let $V$ denote the subspace spanned by $s$ samples drawn independently according to the distribution \eqref{eqn:sampling_distribution}, where each sample has the form $f(A_i)+E_i$ for some $i\in [n]$. Suppose that $\norm{E_i}_2 \leq \gamma\norm{f(A_i)}_2$ for some $\gamma > 0$.  Then with probability at least $9/10$, there exists an orthonormal set of vectors $y_1,y_2,\dots,y_k$ in $V$ such that
\[
\Norm{f(A) - f(A)\sum_{j=1}^k y_j y_j^\top}_F^2 
\leq \min_{D: \rk(D) \leq k}\Norm{f(A) - D}_F^2 + \frac{10k}{sc}(1+\gamma)^2 \Norm{f(A)}_F^2.
\]
\end{theorem}

The theorem shows that the subspace spanned by a sample of columns chosen according to \eqref{eqn:sampling_distribution} contains an approximation to $f(A)$ that is nearly the best possible. Note that if the top $k$ right singular vectors of $S$ belong to this subspace, then $f(A)\sum_{t=1}^k v_t v_t^\top$ would provide the required approximation to $f(A)$ and we would be done. 

Now, the difference between Theorem~\ref{thm:main-sampling} and the assumption~\eqref{eqn:sampling_distribution} is that we do not have control over $p_i$ for an $O(\eps)$-fraction of the rows (in squared row norm contribution) in Theorem~\ref{thm:main-sampling}. Let $A'$ be the submatrix of $A$ after removing those rows, then $\norm{f(A)}_F\leq (1+O(\eps))\norm{f(A')}_F$. We can apply Theorem~\ref{thm:low-rank-raw} to $A'$ and take more samples such that we obtain $s$ rows from $A'$ (which holds with $1-\exp(-\Omega(s))$ probability by a Chernoff bound).
We therefore have the following corollary.

\begin{corollary}\label{cor:low-rank}
Let $y_i$'s be as in Algorithm~\ref{alg:low-rank} and $c$ and $\eps$ be as in Theorem~\ref{thm:main-sampling}. It holds with probability at least $0.7$ that
\[
\Norm{f(A)-f(A)\sum_j y_j y_j^\top}_F^2 
\leq \min_{D: \rk(D) \leq k}\Norm{f(A) - D}_F^2  + \left( \frac{30k}{sc} + \eps\right)\Norm{f(A)}_F^2.
\]
\end{corollary}
\begin{proof}
First, it follows from a Chernoff bound and a union bound that we can guarantee with probability at least $0.9$ that all samples have the form $f(A_i)+E_i$ with small $\norm{E_i}_2$. Then, it follows from another Chernoff bound that with probability at least $0.9$, it holds that there are $s/2$ samples from $A'$. We apply Theorem~\ref{thm:low-rank-raw} to $A'$ and $s/2$ and obtain that
\[
\Norm{f(A')-f(A')\sum_j y_j y_j^\top}_F^2 \leq \min_{D: \rk(D) \leq k}\Norm{f(A') - D}_F^2  + \frac{30k}{sc}\Norm{f(A')}_F^2.
\]
Suppose that $A''$ is the submatrix of $A$ which consists of the rows of $A$ that are not in $A'$. Then $f(A)$ is the (interlacing) concatenation of $f(A')$ and $f(A'')$. Since $\norm{f(A'')}_F^2\leq \eps\norm{f(A)}_F^2$ and $y_1,\dots,y_k$ remains valid if we add more samples,
\begin{align*}
&\quad\ \Norm{f(A)-f(A)\sum_j y_j y_j^\top}_F^2 \\
&= \Norm{f(A')-f(A')\sum_j y_j y_j^\top}_F^2 + \Norm{f(A'')-f(A'')\sum_j y_j y_j^\top}_F^2\\
&\leq \min_{D: \rk(D) \leq k}\Norm{f(A') - D}_F^2 + \frac{30k}{sc}\Norm{f(A)}_F^2 + \Norm{f(A'')}_F^2 \\
&\leq \min_{D: \rk(D) \leq k}\Norm{f(A) - D}_F^2 + \left(\frac{30k}{sc} + \eps\right)\Norm{f(A)}_F^2.
\end{align*}
The overall failure probability combines that of Theorem~\ref{thm:main-sampling}, Theorem~\ref{thm:low-rank-raw} and the events at the beginning of this proof.

For the second result, take $s=O(k/\eps)$ and rescale $\eps$.
\end{proof}

Note that Algorithm~\ref{alg:h-sketch-sample} can be easily modified to return the sampling probability of the sampled column, which is just $\lambda_i^{(j)}/\sum_j\wtilde{M}_j$. However, for each sample, we may lose control of it with a small constant probability. To overcome this, inspecting the proof of $H$-\textsf{Sketch}, we see that for fixed stream downsampling and fixed $\zeta$ in Algorithm~\ref{alg:h-sketch-sample}, repeating each heavy hitter structure $O(\log(L/\delta))$ times and taking the median of each $\wtilde{M}_j$ will lower the failure probability of estimating the contribution of each important level to $\delta/(L+1)$, allowing for a union bound over all levels. Hence, with probability at least $1-\delta$, we can guarantee that we obtain 
a $(1\pm O(\eps))$-approximation to $\norm{(f(A))_u}_2^2$ and thus a $(1\pm O(\eps))$-approximation to $\norm{f(A)}_F^2$. Hence the returned $\hat p_u$ is a $(1\pm O(\eps))$-approximation to the true row-sampling probability $p_u = \norm{(f(A))_u}_2^2/\norm{f(A)}_F^2$. Different runs of the sampling algorithm may produce different values of $\hat p_u$ for the same $u$ but they are all $(1\pm O(\eps))$-approximations to $p_u$. We can guarantee this for all our $s$ samples by setting $\delta=O(1/s)$, which allows for a union bound over all $s$ samples.

Therefore, at the cost of an extra $O(\log s)$ factor in space, we can assume that $\hat p_u = (1\pm O(\eps))p_u$ for all $s$ samples. The overall algorithm is presented in Algorithm~\ref{alg:low-rank}. 

\begin{algorithm}[t]
\caption{Rank-$k$ Approximation} \label{alg:low-rank}
\textbf{Input:} $A \in \mathbb{\R}^{n\times n}$, rank parameter $k$, number of samples $s$
\begin{algorithmic}[1]
	\State Initialize $s$ parallel instances of (modified) Algorithm~\ref{alg:h-sketch-sample}
	\State Let $(h_1, \hat p_1),\dots,(h_q,\hat p_q)$ be the returned vectors and the sampling probability from the $s$ instances of (modified) Algorithm~\ref{alg:h-sketch-sample}
	\State $F\gets$ concatenated matrix $\begin{pmatrix} \frac{h_1}{\sqrt{s\hat p_1}} & \cdots & \frac{h_s}{s\hat p_s}\end{pmatrix}$
	\State Compute the top $k$ left singular vectors of $F$, forming $L\in \R^{n\times k}$
	\State \Return $L$
\end{algorithmic}
\end{algorithm}

The following main theorem follows from Corollary~\ref{cor:low-rank} and the argument in~\cite{fkv04}.

\begin{theorem}
Let $s=O(k/\eps)$ be the number of samples and $y_1,\dots,y_k$ be the output of Algorithm~\ref{alg:low-rank}. It holds with probability at least $0.7$ that
\[
\Norm{f(A) - L L^\top f(A)}_F^2 
\leq \min_{D: \rk(D) \leq k}\Norm{f(A) - D}_F^2  + \eps\Norm{f(A)}_F^2.
\]
\end{theorem}

\section{Experiments}
To demonstrate the benefits of our algorithm empirically, we conducted experiments on  low-rank approximation with a real NLP data set and used the function $f\left( x\right) =\log{\left(\left| x\right| +1 \right) } $. 

The data we use is based on the Wikipedia data used by Liang et al.~\cite{lswyy20}. The data matrix $A'\in\R^{n\times n}$ ($n=10^4$) contains information about the correlation among the $n$ words. Its entries are $A'_{i,j}=p_{j}\log(N_{i,j}N/(N_{i}N_{j})+1)$, where $N_{i,j}$ is the number of times words $i$ and $j$ co-occur in a window size of $10$, $N_{i}$ is the number of times word $i$ appears and $N_{i}$'s are in a decreasing order, $N = \sum_i N_i$  and $p_{j} = \max\{1, (N_j/N_{10})^2 \}$ is a weighting factor which adds larger weights to more frequent words. Since $A'$, and thus $f(A')$, have almost the same column norms, we instead consider $A = A' - \mathbf{1}\mathbf{1}^\top A'$, where $\mathbf{1}$ is the vector of all $1$ entries.

We compare the accuracy and runtime with the previous three-pass algorithm of Liang et al.~\cite{lswyy20}. The task is to find $L\in \R^{n\times k}$ with orthornormal columns to minimize the error ratio
\[
e(L) = \frac{ \Norm{f(A) - LL^\top f(A)}_{F} }{ \Norm{f(A) - UU^\top f(A)}_{F} },
\]
where $U\in \R^{n\times k}$ has the top $k$ left singular vectors of $f(A)$ as columns. The numerator $\norm{f(A) - LL^\top f(A)}_F$ is the approximation error and the denominator $\norm{f(A) - UU^\top f(A)}_F$ is the best approximation error, both in Frobenius norm.

\subsection{Algorithm Implementation}
We present our empirical results for the one-pass algorithm and a faster implementation of the two-pass algorithm in Figure \ref{fig:accuracy} and Figure \ref{fig:time}. Both algorithms run in $0.1\%$ of the runtime of the three-pass algorithm of Liang et al.~\cite{lswyy20}. The one-pass algorithm is less accurate than the two-pass algorithm when the space usage is small, which is not unexpected, because the second pass enables noiseless column samples. Still, as discussed below, one-pass algorithms are essential in certain internet NLP applications. Even our two-pass algorithm has a considerable advantage over the prior three-pass algorithm, by matching its accuracy in significantly less time and one fewer pass. 

For the two-pass algorithm, we sample the columns of $A$ using Algorithm~\ref{alg:h-sketch-sample} in the first pass and only recover the positions of the heavy columns in each magnitude level. Taking $s$ samples will incur $O(s\poly(\eps^{-1}\log n))$ heavy hitters in total. In the second pass we obtain precise $f(A)$ values for our samples and thus noiseless column samples. Then we calculate the sampling probability $\hat p_u$ according to the error-free column norms.

To reduce the runtime, we do not run $s$ independent copies of Algorithm~\ref{alg:h-sketch-sample} for $s$ samples; instead, we take $m$ samples from each single run of Algorithm~\ref{alg:h-sketch-sample} and run $s/m$ independent copies of Algorithm~\ref{alg:h-sketch-sample}. Hence the $s$ samples we obtain are not fully independent. 

All experiments are conducted under MATLAB 2019b on a laptop with a 2.20GHz CPU and 16GB RAM.

We set $k=10$, $m=100$ and plot the error ratios of our algorithm in Figure~\ref{fig:accuracy}. For each value of space usage, the mean and standard deviation are reported from 10 independent runs. In the same figure, we also plot the results of the three-pass algorithm of Liang et al.~\cite{lswyy20} at comparable levels of space usage. Since the three-pass algorithm is considerably slower, we run the three-pass algorithm only once. Additionally we plot the runtimes of all algorithms in Figure~\ref{fig:time}. 

We can observe that even at the space usage of approximately $12\%$ of the input data, the error ratio of our two-pass algorithm is stably around $1.05$. The one-pass algorithm is less accurate than the one-pass algorithm when the space usage is less than $20\%$, which is not unexpected, because the second pass enables noiseless column samples. Overall, the error ratio of the one- and two-pass algorithms is similar to that of the three-pass algorithm for space usage level at least $0.2$, while the runtime of both algorithms is at most $0.1\%$ of that of the three-pass algorithm, which is a significant improvement.

\begin{figure}
    \centering
    \begin{tikzpicture}
\begin{axis}[ymin = 0.95, ymax = 1.35, xmin = 0, xmax = 0.4, 
             xlabel={space usage}, ylabel={error ratio},
             tick label style={/pgf/number format/fixed},
             width = 10cm, height = 5cm]
\addplot[color=teal, mark=triangle,]
 plot [error bars/.cd, y dir = both, y explicit]
 table[row sep=crcr, x index=0, y index=1, y error index=2]{
0.3570 1.031855 0.035195 \\
0.2940 1.031427 0.019557 \\
0.2380 1.026453 0.019931 \\
0.1760 1.059894 0.078924 \\
0.1200 1.104430 0.084912 \\
0.0630 1.201258 0.137502 \\
};
\addplot[color=blue, mark=*,]
 plot [error bars/.cd, y dir = both, y explicit]
 table[row sep=crcr, x index=0, y index=1, y error index=2]{
0.3564 1.030569 0.019294 \\
0.297  1.029355 0.012112 \\
0.2376 1.032825 0.033609 \\
0.1782 1.041561 0.020620 \\
0.1188 1.050296 0.022213 \\
0.0594 1.149146 0.201046 \\
};
\addplot[color=magenta, mark=o,]
 plot [error bars/.cd, y dir = both, y explicit]
 table[row sep=crcr, x index=0, y index=1]{
0.3564 1.034684 \\
0.297 1.034686  \\
0.2376 1.034694 \\
0.1782 1.038564 \\
0.1188 1.038539 \\
0.0594 1.032188  \\
};
\addlegendentry{one-pass alg.};
\addlegendentry{two-pass alg.};
\addlegendentry{three-pass alg.};
\end{axis}
\end{tikzpicture}
    \vspace{-0.5cm}
    \caption{Error ratios of the one-pass, two-pass and three-pass algorithms. The $x$-axis is the ratio between the space of the sketch maintained by the tested algorithm and the space to store the input matrix. The $y$-axis is the error ratio $e(L)$. Solid dots denote the mean of the error ratios over $10$ independent trials and the vertical bars denote the standard deviation of the one-pass and two-pass algorithms.}
    \label{fig:accuracy}
\end{figure}
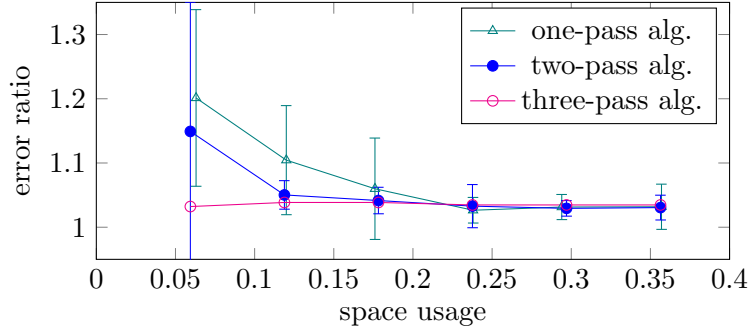

    \pgfplotsset{
    every non boxed x axis/.style={} 
}
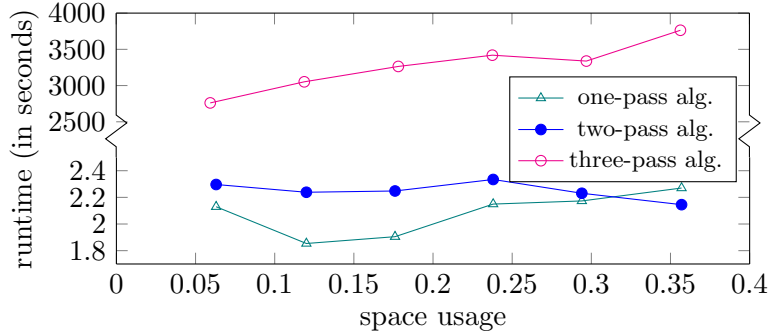
\begin{figure}
    \centering
    \begin{tikzpicture}
    \begin{groupplot}[
    group style={
        group name=runtime,
        group size=1 by 2,
        vertical sep=0pt,
        xticklabels at=edge bottom
    },
    xmin=0, xmax=0.4,
    width = 10cm,
    tick label style={/pgf/number format/.cd, set thousands separator={}, fixed},
    legend style={at={(0.8,1.76)},anchor=north},
    legend style={font=\footnotesize},
    groupplot ylabel = {runtime (in seconds)},
    ]
    \nextgroupplot[ymin = 2010, ymax = 4000, axis x line=top,
                    axis y discontinuity=crunch, height = 3.5cm]
\addplot[color=magenta, mark=o]
 plot [error bars/.cd, y dir = both, y explicit]
 table[row sep=crcr, x index=0, y index=1]{
0.3564 3763.1446 \\
0.297  3338.0686  \\
0.2376 3419.3905 \\
0.1782 3263.936 \\
0.1188 3051.819 \\
0.0594 2760.607  \\
};\label{plots:3pass_runtime}
    \nextgroupplot[ymin = 1.7, ymax = 2.5, axis x line=bottom, xlabel={space usage},
                height = 3cm]
\addplot[color=teal, mark = triangle]
 plot [error bars/.cd, y dir = both, y explicit]
 table[row sep=crcr, x index=0, y index=1]{
0.3570 2.270008  0.103672 \\
0.2940 2.173176  0.054678 \\
0.2380 2.149617  0.080564 \\
0.1760 1.904915  0.080979 \\
0.1200 1.853819  0.045922 \\
0.0630 2.129655  0.042163 \\
};\label{plots:1pass_runtime}
\addplot[color=blue, mark = *]
 plot [error bars/.cd, y dir = both, y explicit]
 table[row sep=crcr, x index=0, y index=1]{
0.3570  2.145211  0.092794 \\
0.2940  2.230715  0.222024 \\
0.2380  2.334207  0.374474 \\
0.1760  2.247841  0.207320 \\
0.1200  2.238326  0.158536 \\
0.0630  2.296440  0.160212 \\
};\label{plots:2pass_runtime}
\addlegendentry{one-pass alg.}
\addlegendentry{two-pass alg.};
\addlegendimage{/pgfplots/refstyle=plots:3pass_runtime,mark=o}\addlegendentry{three-pass alg.}
    \end{groupplot}
    \pgfresetboundingbox\path (runtime c1r1.outer north west) rectangle (runtime c1r2.outer south east); 
\end{tikzpicture}
    \vspace{-0.5cm}
    \caption{Runtimes of the one-pass, two-pass and three-pass algorithms. The $x$-axis is the ratio between the space of the sketch maintained by the tested algorithm and the space to store the input matrix. The $y$-axis is the runtime in seconds. Solid dots denote the mean of $10$ independent trials of the one-pass and two-pass algorithms. The standard deviations are all less than $1.5$ seconds and thus omitted. }
    \label{fig:time}
\end{figure}

\section{Application to Linear Regression}\label{sec:regression}

In this section, we consider approximately solving the linear regression problem using the $H$-Sketch from Section~\ref{sec:h-sketch}. 

We shall need to sample rows from the concatenated matrix $Q = \begin{pmatrix} f(A) & b\end{pmatrix}$, where $A\in \R^{n\times d}$ and $b\in \R^n$ are given in a turnstile stream, and $f(x) = \ln(1+|x|)$ is the transformation function. This can still be achieved using the same $H$-\textsf{Sketch} in Section~\ref{sec:h-sketch}, applied to the concatenated matrix $\begin{pmatrix}A & b\end{pmatrix}$, with the transformation function $f(x)$ $(x\in\R^{d+1})$ replaced with $g(x) = \begin{pmatrix}f(x_{1:d}) & x_{d+1}\end{pmatrix}$, where $x_{1:d}$ denotes the first $d$ coordinate of $x$ and $x_{d+1}$ the last coordinate of $x$.
Then, using identical arguments in Section~\ref{sec:h-sketch} for the squared $\ell_2$-norm on the first $d$ coordinates and a standard \textsf{Count-Sketch} argument for the last coordinate, it is straightforward to show an analogous version of Theorem~\ref{thm:main-sampling} as below. We omit the identical proof. 
\begin{theorem}\label{thm:concatenated_sampling}
Let $\eps>0$, and let $H_{j_0,i_0}$ be the vector returned by Algorithm~\ref{alg:h-sketch-sample}. With probability at least $0.9$, it holds that there exists $u\in [n]$ such that 
\[
(1-O(\eps))\left(\Norm{f(A_u)}_2^2 + \Abs{b_u}^2 \right) \leq \Norm{g(H_{j_0,i_0})}_2^2 
 \leq (1+O(\eps))\left(\Norm{f(A_u)}_2^2 + \Abs{b_u}^2\right).
\]
\end{theorem}

Theorem~\ref{thm:concatenated_sampling} states that each sample is a noisy version of the $u$-th row of $Q$. Let $p_u = \norm{Q_u}_2^2/\norm{Q}_F^2$ be the true sampling probability of the $u$-th row. 
As argued at the beginning at Section~\ref{sec:low-rank}, we may assume, at the cost of an $O(\log s)$ factor in space, that every sample is good, i.e., the returned sampling probability $\hat{p}_u$ satisfies that $\hat{p}_u = (1\pm O(\eps))p_u$ and the noise in each sample is at most an $O(\eps)$-fraction in its $\ell_2$ norm.

Below we present our algorithm for linear regression in Algorithm~\ref{alg:linear-regression}, assuming that every sample is good in the sense that $\hat p_u = (1\pm O(\eps))p_u$ and the noise in each sample is at most an $O(\eps)$-fraction in $\ell_2$-norm. The guarantee is given in Theorem~\ref{thm:linear-regression} and the proof is deferred to Section~\ref{sec:regression_proof}.

\begin{algorithm}[tb]
\caption{Linear Regression}\label{alg:linear-regression}
\begin{algorithmic}[1]
    \Require $A\in \mathbb{R}^{n \times d}$, number of samples $s$
    \State Initialize $s$ parallel instances of Algorithm~\ref{alg:h-sketch-sample} 
    \State Run Algorithm~\ref{alg:h-sketch-sample} on the concatenated matrix $\begin{pmatrix}f(A) & b\end{pmatrix}$ to obtain $s$ row samples $h_1,\dots,h_s$ and the corresponding sampling probabilities $\hat p_1,\dots,\hat p_s$
	\State $T\gets$ vertical concatenation of $\frac{h_1}{\sqrt{s\hat p_1}},\dots, \frac{h_s}{\sqrt{s\hat p_s}}$
	\State $\tilde{M}\gets$ first $d$ columns of $T$, $\tilde{b}\gets$ last column of $T$
    \State $\tilde{x} \gets \arg \min_{x\in \mathbb{R}^d} \norm{\tilde{M}x-\tilde{b}}_2$
\end{algorithmic}
\end{algorithm}

\begin{theorem}\label{thm:linear-regression}
Given matrix $A\in \mathbb{R}^{n\times d}$ and vector $b\in \mathbb{R}^{n}$, let $\kappa=\kappa(f(A))$ be the condition number of the transformed matrix. Let $s =  O(\frac{d\kappa^2}{\eps^2} \log\frac{d}{\delta})$, then Algorithm~\ref{alg:linear-regression} outputs a vector $\tilde{x}\in \mathbb{R}^d$, which, with probability at least $1-\delta$, satisfies 
\[
    \Norm{\tilde{M}\tilde{x}-\tilde{b}}_2 \leq (1+\eps)\min_{x\in \mathbb{R}^d}\Norm{f(A)x-b}_2 +\Delta,
\]
where 
\[
\Delta = \eps \left(\sqrt{d+\frac{\norm{b}_2^2}{\norm{f(A)}_2^2}}\kappa\Norm{b}_2 + \sqrt{\norm{f(A)}_F^2+\norm{b}_2^2} \right).
\] The total space used by Algorithm~\ref{alg:linear-regression} is $\tilde{O}(d^2\kappa^2\log\frac{1}{\delta})\cdot \poly(\log n,\frac{1}{\eps})$.
\end{theorem} 

Finally, we note an error in~\cite{lswyy20}. Let $G=f(A)$ and $S$ be a subspace embedding sketching matrix of $s=\poly(d/\eps)$ rows for the column space of $G$. In their proof of the regression problem in Section E, the upper bound of $C_2$ is wrong as it claims that $\norm{(\wtilde{SG})^\dagger}_2\leq C\norm{(SG)^\dagger}_2$; a correct bound should be $\norm{(\wtilde{SG})^\dagger}_2 \leq 10 nd\kappa\eta\norm{(SG)^\dagger}_2$, where $\eta = \max_{i,j} \abs{(\wtilde{SG})_{i,j} - (SG)_{i,j}}$ could be linear in $\sqrt{n}$ by their \textsc{LogSum} guarantee. The proof in~\cite{lswyy20} does not account for such a dependence on $n$ and $\kappa$. Correcting the proof would yield a similar bound as ours but with an addtive error $\Delta = n^{2} \kappa^3 \poly(\frac{\eps}{d}) \norm{b}_2$, which depends polynomially on $n$. Our additive error has no dependence on $n$ but depends on $\norm{b}_2^2/\norm{f(A)}_2^2$ and has an additional additive term of $\eps\norm{Q}_2$, which is an artefact of sampling the rows of $Q$.

\section*{Acknowledgements}
Y.~Li was partially supported and Y.~Sun was supported by Singapore Ministry of Education (AcRF) Tier 2 grant MOE2018-T2-1-013. D. Woodruff thanks support from NSF grant No. CCF-1815840, Office of Naval Research grant N00014-18-1-256, and a Simons Investigator Award.

\bibliographystyle{plain}
\bibliography{reference}

\appendix

\section{Omitted Proofs of Useful Inequalities}\label{sec:ineq proofs}
\subsection{Proof of Proposition~\ref{prop:log_ineq}}
\begin{proof}
Let $h(x) = (1+2x)\ln(1+x) - x$. Since $h(0) = 0$, it suffices to show that $h'(x) > 0$. We calculate that
\[
h'(x) = \frac{x}{1+x} + 2\ln(1+x).
\]
Since $h'(0)=0$, it suffices to show that $h''(x) > 0$. This can be readily verified by calculating that
\[
h''(x) = \frac{3+2x}{(1+x)^2} > 0.\qedhere
\]
\end{proof}

\subsection{Proof of Proposition~\ref{prop:log^2}}
\begin{proof}
Let $f(x,y) = \ln^2(1+x) + \ln^2(1+y) - \ln^2(1+\sqrt{x^2+y^2})$. It suffices to show that $f(x,y)\geq 0$. The inequality is clearly true when $x=0$ or $y=0$. Note that
\begin{align*}
\frac{\partial f}{\partial x} &= 2\left(\frac{\log(1+x)}{1+x} - \frac{x\ln(1+\sqrt{x^2+y^2})}{x^2+y^2+\sqrt{x^2+y^2}}\right)\\
\frac{\partial f}{\partial y} &= 2\left(\frac{\log(1+y)}{1+y} - \frac{y\ln(1+\sqrt{x^2+y^2})}{x^2+y^2+\sqrt{x^2+y^2}}\right)
\end{align*}
Assuming $x,y>0$, $\partial f/\partial x = \partial f/\partial y = 0$ implies that
\[
\frac{\log(1+x)}{x(1+x)} = \frac{\log(1+x)}{y(1+y)}.
\]
It is easy to verify that $\log(1+x))/(x(1+x))$ is decreasing w.r.t.\ $x$ (checking the derivative and using Proposition~\ref{prop:log^2}), so we must have $x=y$. Now, let
\[
h(x) = \frac{\partial f}{\partial x}(x,x) = \frac{2\ln(1+x)}{1+x} - \frac{\sqrt{2}\ln(1+\sqrt{2}x)}{1+\sqrt{2}x}.
\]
We shall show that $h(x) > 0$ for all $x > 0$. This will imply that $f(x,y)$ has no local minimum or maximum when $x,y>0$ and so it is easy to see that $f(x,y)$ attains the minimum at its boundary $x=0$ or $y=0$, yielding that $f(x,y)\geq 0$ for all $x,y\geq 0$.

To see that $h(x) > 0$, let
\[
g(a) = \frac{\ln(1+ax)}{a(1+ax)}.
\]
We calculate 
\[
g'(a) = \frac{ax - (1+2ax)\ln(1+ax)}{a^2(1+ax)^2}.
\]
It follows from Proposition~\ref{prop:log_ineq} that $g'(a) < 0$. Hence $g(a)$ is decreasing w.r.t.\ $a$ and $g(\sqrt{2})<g(1)$, which is exactly $\frac{1}{\sqrt 2}h(x) > 0$.
\end{proof}

\subsection{Proof of Lemma~\ref{lem:light_noise_f}}
\begin{proof}
It is clear that the base of the logarithm does not matter and we assume that the base is $e$. Let $Z = \sum_i \eps_i a_i$ and $\sigma^2 = \sum a_i^2$. Then $\E Z^2 = \sigma^2$ and $\E|Z|\leq (\E|Z|^2)^{1/2} = \sigma$. Let $g(x) = \ln(1+x)$ and
\[
Z_1 = \begin{cases}
		|Z|, & |Z| \geq e-1;\\
		0, & \text{otherwise},
	\end{cases} \quad	
Z_2 = \begin{cases}
		0, & |Z| \geq e-1;\\
		|Z|, & \text{otherwise}.
	\end{cases}
\]
Then $|Z| = Z_1 + Z_2$ and
\[
\E g(|Z|)^2 = \E(g(Z_1+Z_2))^2 \leq \E (g(Z_1)+g(Z_2))^2 \leq \E 2(g(Z_1)^2+g(Z_2)^2),
\]
where the first inequality follows from Proposition~\ref{lem:log_additive}. For the first term, we define $h(x) = g(x)\cdot \mathbf{1}_{\{x\geq e-1\}}$. Then $h(x)^2$ is concave on $[0,\infty)$. Hence
\[
\E g(Z_1)^2 = \E h(Z_1)^2 = \E h(|Z|)^2 \leq h(\E|Z|)^2 \leq h(\sigma)^2 \leq g(\sigma)^2.
\]
Next we upper bound the second term. The first case is $\sigma\leq e-1$. Since $\E Z^4\leq 3\sigma^4$, it holds that $\Pr\{Z_2\geq t\sigma\} \leq \Pr\{|Z|\geq t\sigma\}\leq 3/t^4$. Then
\begin{align*}
\E g(Z_2)^2 &\leq \E g(e-1)g(Z_2)\\ &= \E g(Z_2)\\ &= \int_0^{e-1} g(x) \Pr\{Z_2\geq x\} dx \\
&= \sigma\int_0^{(e-1)/\sigma} g(t\sigma) \Pr\{Z_2\geq t\sigma\} dt \\
&= \sigma^2\int_0^{(e-1)/\sigma} g(t) \Pr\{Z_2\geq t\sigma\} dt \quad (\text{by Proposition~\ref{lem:log_additive}})\\
&\leq \sigma^2 \left(\int_0^1 g(t) dt + 3\int_{1}^{(e-1)/\sigma} \frac{g(t)}{t^4} dt\right) \\
&\leq C_1 \sigma^2\\
&\leq C_1 (e-1)^2 g(\sigma)^2,
\end{align*}
where $C_1 > 0$ is an absolute constant and the last inequality follows from the fact that $g(x)\geq x/(e-1)$ on $[0,e-1]$. The second case is $\sigma > e-1$. In this case, 
\[
\E g(Z_2)^2\leq 1 \leq g(\sigma)^2.
\]
Therefore, we conclude that
\[
\E g(|Z|)^2 \leq 2(1 + C_1(e-1)^2) g(\sigma)^2 = C_2 g\left(\sqrt{\sum_i a_i^2}\right)^2 \leq C_2 \sum_i g(|a_i|)^2,
\]
where the last inequality follows from Proposition~\ref{prop:log^2}.
\end{proof}

\subsection{Proof of Lemma~\ref{lem:add_noise}}
\begin{proof}
We first prove the upper bound.
\begin{align*}
\Norm{f(y+z)} _{2}^{2} &= \sum_i f(y_i + z_i)^2 \\
&\le \sum_{i} \left[ f(y_i) + f(z_i) \right] ^{2} \quad (\text{Proposition~\ref{lem:log_additive}})\\
&= \sum_{i} f(y_i) ^{2} + \sum_{i}f(z_i) ^{2}+\sum_{i}2f(y_i)f(z_i) \nonumber \\
&\le \sum_{i} f(y_i) ^{2} + \xi^{2}\sum_{i}f(y_i) ^{2} + 2 \sqrt{\sum_{i} f(y_i)^{2}}  \sqrt{\sum_{i} f(z_i) ^{2}}  \quad \text{(Cauchy-Schwarz)} \\
&\le \left(\xi^{2}+2\xi +1 \right) \Norm{f(y)}_{2}^{2}\\
&\leq (1+3\xi) \Norm{f(y)}_2^2. \quad (\text{since }\xi<1) 
\end{align*}

Next we prove the lower bound. Let $I = \{i: y_i z_i\leq 0\}$, $J_1 = \{i\in I: |y_i| \leq |z_i|\}$ and $J_2 = \{i\in I: |z_i| < |y_i|\leq \zeta^{-1}|z_i|\}$ for some $\zeta<1$ to be determined. Then
\begin{align*}
\Norm{f(y+z)}_2^2 &= \sum_{i\in J_1} f(y_i+z_i)^2 + \sum_{i\in J_2} f(y_i+z_i)^2  + \sum_{i\in I\setminus (J_1\cup J_2)} f(y_i+z_i)^2 + \sum_{i\not\in I} f(y_i+z_i)^2 \\
&\geq \sum_{i\in I\setminus (J_1\cup J_2)} f(y_i + z_i)^2 + \sum_{i\not\in I} f(y_i)^2.
\end{align*}
When $i\in I\setminus (J_1\cup J_2)$, we have $|z_i|\leq \zeta|y_i|$.  It then follows that
\[
\log(|y_i+z_i|+1) \geq \log((1-\zeta)|y_i|+1)\geq (1-\zeta)\log(|y_i|+1),
\]
where, for the last inequality, one can easily verify that $h_\eps(x) = \frac{\log(1+(1-\eps)x)}{\log(1+x)}$ is increasing on $[0,\infty)$ and $\lim_{x\to 0^+} h_\eps(x) = 1-\eps$. Hence
\[
\sum_i f(y_i+z_i)^2 \geq (1-\zeta)^2 \sum_{i\in I\setminus (J_1\cup J_2)} f(y_i)^2 + \sum_{i\not\in I} f(y_i)^2 \geq (1-\zeta)^2\sum_{i\not\in J_1\cup J_2} f(y_i)^2.
\]
Now, note that
\[
\sum_{i\in J_1} f(y_i)^2 \leq \sum_{i\in J_1} f(z_i)^2 \leq \|f(z)\|_2^2 \leq \xi^2 \Norm{f(y)}_2^2
\]
and (using Proposition~\ref{lem:log_additive})
\[
\sum_{i\in J_2} f(y_i)^2 \leq \zeta^{-2} \sum_{i\in J_1} f(z_i)^2 \leq \zeta^{-2} \|f(z)\|_2^2 \leq (\zeta^{-1}\xi)^2 \Norm{f(y)}_2^2.
\]
It follows that
\begin{align*}
\sum_i f(y_i+z_i)^2 &\geq (1-\zeta)^2\left(\Norm{f(y)}_2^2 - \xi^2\Norm{f(y)}_2^2 - (\zeta^{-1}\xi)^2\Norm{f(y)}_2^2 \right) \\
&= (1-\zeta)^2(1-\xi^2-(\zeta^{-1}\xi)^2)\Norm{f(y)}_2^2.
\end{align*}
Choosing $\zeta = (\xi^2/(1-\xi^2))^{1/3}$ maximizes the right-hand side, yielding
\[
\Norm{f(y+z)}_2^2\geq (1-3\xi^{2/3})\Norm{f(y)}_2^2. \qedhere
\]
\end{proof}

\section{Obtaining an Overestimate $\widehat{M}$}\label{sec:1-pass-tractable}
In this subsection we verify that $g(x) = \ln^2(1+\eta x)$ is slow-jumping, slow-dropping, and predictable, where the three properties are defined in~\cite{BCWY16}.

To show that $g$ is slow-jumping, we shall verify that for any $\alpha > 0$, $g(y) \leq \lfloor \frac{y}{x}\rfloor^{2+\alpha} x^\alpha g(x)$ for all $x<y$, whenever $y$ is sufficiently large. (i) When $x \geq y/2$, it suffices to show that $g(y)\leq x^\alpha g(x)$. Since $g(x)$ is increasing, it reduces to showing $g(y)\leq (y/2)^\alpha g(y/2)$. This clearly holds for all large $y$ because one can easily check that $\ln(1+y) \leq 2\ln(1+\frac{y}{2})$ when $y>0$. (ii) When $x < y/2$, we shall show that $g(y)\leq (\frac{y}{x}-1)^{2+\alpha}x^\alpha g(x)$, i.e., $g(y)\leq (\frac{y-x}{x})^2 (y-x)^{\alpha} g(x)$. Since $x<y/2$, we have $y-x\geq y/2$ and thus it suffices to show that $g(y) \leq \frac{1}{4}(\frac{y}{x})^2(\frac{y}{2})^\alpha g(x)$, and for large $y$ that $\frac{g(y)}{y^2}\leq \frac{g(x)}{x^2}$, which can be easily verified. This concludes the proof that $g$ is slow-jumping.

To show that $g$ is slow-dropping, we shall verify that for any $\alpha>0$ it holds that $g(y)\geq g(x)/x^\alpha$ for all $x<y$ whenever $y$ is sufficiently large. This holds obviously because $g(x)$ is increasing.

To show that $g$ is predictable, we shall verify that for any $\gamma \in (0,1)$ and subpolynomial $\eps(x)$, it holds that $g(y)\geq x^{-\gamma}g(x)$ for all sufficiently large $x$ and all $y\in [1,x^{1-\gamma}]$ such that $g(x+y)>(1+\eps(x))g(x)$. This holds automatically because $g(2x)/g(x)\to 1$ as $x\to\infty$ and thus for any given $\eps(x)$,when $x$ is sufficiently large, it  would not hold that $g(x+y)>(1+\eps(x))g(x)$ for $y\in [1,x]$.

\section{Proof of Theorem~\ref{thm:low-rank-raw}}\label{sec:proof of low rank raw}
\begin{proof}
For notational convenience, let $G = f(A)$. Let $S$ be a random sample of $s$ rows chosen from a distribution that satisfies \eqref{eqn:sampling_distribution}. We can write the $i$-th sample as $G_i+E_i$ for some error vector $E_i$. Consider the singular value decomposition of $G = \sum_t \sigma_t u_t v_t^\top$.

For each $t$, we define a random vector
\begin{equation*}
w_t = \frac{1}{s}\sum_{i\in S}\frac{(u_t)_i}{p_i}(G_i + E_i).
\end{equation*}
Note that $S$ in general consists of sampled columns of $f(A)$ with noise. The vectors $w_t$ are clearly in the subspace generated by $S$. We first compute $\E w_t$. We can view $w_t$ as the average of $s$ i.i.d. random variables $X_1,\dots,X_s$, where each $X_j$ has the following distribution:
\[
X_j=\frac{(u_t)_i}{p_i}(G_i+E_i)\text{ with probability }p_i, \quad  i=1,2,\dots n.
\]
Taking expectations,
\[
    \E X_j = \sum_{i=1}^n \frac{(u_t)_i}{p_i}(G_i + E_i)p_i = u_t^\top (G+E) = \sigma_t v_t^\top + u_t^\top E
\]
Hence
\[
\E w_t = \E X_j = \sigma_t v_t^\top + u_t^\top E
\]
and
\[
\Norm{\E X_j}_2^2 = \sigma_t^2 + 2\langle \sigma_t v_t^\top, u_t^\top E\rangle + \Norm{u_t^\top E}_2^2 \leq \sigma_t^2 + 2\langle \sigma_t v_t^\top, u_t^\top E\rangle + \Norm{E}_2^2.
\]
We also calculate that
\begin{align*}
\E \Norm{X_j}_2^2 &= \sum_i \frac{(u_t)_i^2}{p_i^2} \Norm{G_i+E_i}_2^2 \cdot p_i  \\
                  &\leq \sum_i \frac{(u_t)_i^2}{p_i} (\Norm{G_i}_2 + \Norm{E_i}_2)^2 \\
                  &\leq \sum_i (u_t)_i^2 \frac{\norm{G}_F^2}{c\Norm{G_i}_2^2} (1+\gamma)^2\Norm{G_i}_2^2 \\
                  &= \frac{(1+\gamma)^2}{c}\Norm{G}_F^2,
\end{align*}
where we used the assumption \eqref{eqn:sampling_distribution} in the third line and the fact that $\norm{u_t}_2=1$ in the last line. It follows that
\begin{align*}
\E\Norm{w_t}_2^2 &= \E\Norm{\frac{1}{s}\sum_j X_j}_2^2 = \frac{1}{s}\sum_j \E\Norm{X_j}_2^2 + \frac{1}{s^2}\sum_{j\neq \ell} \langle \E X_j,\E X_\ell\rangle \\
&\leq \frac{(1+\gamma)^2}{sc}\Norm{G}_F^2 + \frac{s(s-1)}{s^2} \left(\sigma_t^2 + 2\langle \sigma_t v_t^\top, u_t^\top E\rangle + \Norm{E}_2^2\right),
\end{align*}
and thus
\begin{equation}\label{eqn:low-rank-aux1}
\begin{aligned}
\E \Norm{w_t-\sigma_t v_t^\top}_2^2  &= \E \Norm{w_t}_2^2 - 2\langle \E w_t,\sigma_t v_t^\top\rangle +  \sigma_t^2  \\
& \leq \frac{(1+\gamma)^2}{sc}\Norm{G}_F^2 + \sigma_t^2 + 2\langle \sigma_t v_t^\top, u_t^\top E\rangle + \Norm{E}_2^2  - 2\sigma_t^2 - 2\langle u_t^T E,\sigma_t v_t^\top\rangle +  \sigma_t^2 \\
& = \frac{(1+\gamma)^2}{sc}\Norm{G}_F^2.
\end{aligned}
\end{equation}
If $w_t$ were exactly equal to $\sigma_t v_t^\top$ (instead of just in expectation), we would have
\[
G\sum_{t=1}^k v_t v_t^\top = G \sum_{t=1}^k w_t^\top w_t,
\]
which would be sufficient to prove the theorem. We wish to carry this out approximately. To this end, define $\hat y_t =\frac{1}{\sigma_t}w_t^\top$ for $t = 1,2,\dots,s$ and let $V_1 = \spn(\hat y_1, \hat y_2,\dots, \hat y_s) \subseteq V$. Let $y_1,y_2,\dots,y_n$ be an orthonormal basis of $\R^n$ with $V_1 = \spn(y_1,y_2,\dots,y_l)$, where $l=\dim(V_1)$. Let
\begin{equation*}
B=\sum_{t=1}^l G y_t y_t^\top \quad\mbox{and} \quad\hat B=\sum_{t=1}^k G v_t\hat y_t^\top.
\end{equation*}
The matrix $B$ will be our candidate approximation to $G$ in the span of $S$. We shall bound its error using $\hat B$. Note that for any $i \leq k$ and $j > l$, we have $({\hat y_i})^\top y_j = 0$. Thus,
\begin{equation}\label{eqn:B_and_B_hat}
\Norm{G-B}_F^2 =\sum_{i=1}^n \Norm{(G-B)y^{(i)}}_2^2 
 = \sum_{i=l+1}^n \Norm{G y^{(i)}}_2^2 
 = \sum_{i=l+1}^n \Norm{(G-\hat B)y^{(i)}}_2^2 
 \leq \Norm{G -\hat B}_F^2.
\end{equation}
Also,
\[
\norm{G-\hat B}_F^2 = \sum_{i=1}^n \Norm{u_i^\top(G-\hat B)}_2^2 
= \sum_{i=1}^k \Norm{\sigma_i v_i^\top - w_i}_2^2 + \sum_{i=k+1}^n \sigma_i^2
\]
Taking expectations and using~\eqref{eqn:low-rank-aux1}, we obtain that
\begin{equation}\label{eqn:low-rank-aux2}
\E \Norm{G-\hat B}_F^2 \leq \sum_{i=k+1}^n \sigma_i^2 + \frac{k(1+\gamma)^2}{sc}\Norm{G}_F^2.
\end{equation}
Note that $\hat B$ is of rank at most $k$ and $D_k$ is the best rank-$k$ approximation to $G$. We have
\[
\Norm{G-\hat B}_F^2\geq \Norm{G-D_k}_F^2=\sum_{i=k+1}^n\sigma_i^2
\]
Thus $\|G-\hat B\|_F^2-\|G-D_k\|_F^2$ is a non-negative random variable. It follows from \eqref{eqn:low-rank-aux2} that
\[
\Pr\left\{ \Norm{G-\hat B}_F^2 - \Norm{G-D_k}_F^2\geq \frac{10k(1+\gamma)^2}{sc}\Norm{G}_F^2\right\} \leq \frac{1}{10}.
\]
The result follows from \eqref{eqn:B_and_B_hat} and the fact that $\Norm{E}_F^2\leq \gamma\Norm{G}_F^2$.
\end{proof}

\section{Proof of Theorem~\ref{thm:linear-regression}} \label{sec:regression_proof}

By Theorem~\ref{thm:main-sampling}, for every $i\in [s]$, there exists $j(i)$ such that $h_i=(f(A)_{j(i)},b_{j(i)})+F_{j(i)}$, where $F_i=\frac{E_i}{\sqrt{sp_i}}$. We define a new matrix $S$ such that in the $i$-th row of $S$, $S_{i,j(i)}=\frac{1}{\sqrt{sp_{j(i)}}}$ and the other entries are zero. By Theorem~\ref{thm:concatenated_sampling}, we have that the row-sampling probability we use is a $(1\pm O(\eps))$ approximation to the true sampling probability. Therefore, we define matrix $\hat S$ such that in the $i$-th row of $\hat S$, $\hat S_{i,j(i)}=\frac{1}{\sqrt{s\hat p_{j(i)}}}$ and the other entries are zero, and matrix $\hat F$ is such that $\hat F_i = \frac{E_i}{\sqrt{s \hat p_i}}$.  Then, we find that $\hat S\begin{pmatrix}f(A) & b\end{pmatrix}+\hat F=T$.

\begin{proof}
For notational convenience, we let $G=f(A)$ with singular value decomposition $G=U\Sigma V^\top$.
We shall show that $\Norm{I_d - (\hat SU)^\top (\hat SU)}_2$ is small, for which we first show $\Norm{I_d - (SU)^\top (SU)}_2$ is small.

Let $X_i=I_d-Y_i^TY_i$ and $Y_i = \frac{U_{j(i)}}{\sqrt{p_{j(i)}}}$, where $U_t$ is the $t$-th row of $U$,  which means that the $j(i)$-th row of $M$ is chosen in the $i$-th trial. Since
\[
\E (X_i) = I_d - \E (Y_i^TY_i) = I_d - \sum_{t=1}^n p_t \frac{U_t^T}{\sqrt{p_t}} \frac{U_t}{\sqrt{p_t}}=I_d - \sum_{t=1}^n U_t^T U_t = 0,
\]
we can apply Lemma~\ref{lem:matrix-bound} to $X_1,\dots,X_s$, for which we need to upper bound $\norm{X_i}_2$ and $\norm{\E(X_i^2)}_2$.

We first bound $\norm{X_i}_2$.
\[
\Norm{X_i}_2 = \Norm{I_d - Y_i^\top Y_i}_2 \leq 1+\frac{\Norm{U_i^\top U_i}_2}{p_i} \leq 1+\frac{\Norm{U_i}_2^2}{c\Norm{G_i}_2^2}\Norm{G}_F^2 \leq 1+\frac{\sigma_1^2+\cdots +\sigma_d^2}{c\sigma_d^2}\leq 1+\frac{d\kappa^2}{c},
\]
where $\sigma_1\geq\cdots\geq\sigma_d$ are the singular values of $G$, and in the penultimate inequality we use the fact that
$\Norm{G_i}_2=\Norm{U_i\Sigma V^T}_2=\Norm{U_i\Sigma}_2 \geq \sigma_d\Norm{U_i}_2$.

Next, we bound $\Norm{\E (X_i^2)}_2$. Observe that
\begin{align*}
\E (X_i^2+ I_d) &= I_d +\E (I_d - Y_i^\top Y_i)(I_d - Y_i^\top Y_i)=I_d + \E (I_d - 2Y_i^\top Y_i+Y_i^\top Y_iY_i^\top Y_i)\\
&= 2I_d-\E (Y_i^\top Y_i)+\E (Y_i^\top Y_i\Norm{Y_i}_2^2)=\E \left(\frac{\Norm{U_{j(i)}}_2^2}{p_{j(i)}}Y_i^\top Y_i\right),
\end{align*}
and thus
\[
\Norm{\E (X_i^2+I_d)}_2=\Norm{\E \left(\frac{\Norm{U_{j(i)}}_2^2}{p_{j(i)}} Y_i^\top Y_i \right)}_2 \leq \Norm{\E \left(\frac{\Norm{U_i}_2^2}{c\Norm{G_i}_2^2}\Norm{G}_F^2Y_i^\top Y_i \right)}_2 \leq \Norm{\E \left(\frac{d\kappa^2}{c}Y_i^\top Y_i\right)}_2 = \frac{d\kappa^2}{c}.
\]
It follows immediately from the triangle inequality that
\begin{align*}
\Norm{\E X_i^2}_2 \leq \Norm{\E (X_i^2+I_d)}_2 + \Norm{I_d}_2 \leq \frac{d\kappa^2}{c}+1.
\end{align*}
Invoking Lemma~\ref{lem:matrix-bound}, for
\begin{align*}
W=\frac{1}{s}\sum_{i=1}^s X_i=I_d-\frac{1}{s}\sum_{i=1}^s Y_i^\top Y_i=I_d-(SU)^\top (SU),
\end{align*}
and $\rho = \sigma^2 = 1+d\kappa^2/c$, we have that
\[
\Pr\left\{\Norm{I_d -(SU)^\top(SU)}_2 > \eps\right\} \leq 2d\exp\left(-\frac{\eps^2 s}{\sigma^2 + \rho\eps/3}\right) \leq 2d\exp\left( -\frac{\eps^2 s}{2d\kappa^2/c} \right) \leq \delta
\]
by our choice of $s$. Equivalently, with probability at least $1-\delta$, it holds that $\Norm{I_d -(SU)^\top(SU)}_2 \leq \eps$, which implies that $\Norm{SGx}_2=(1\pm\eps)\Norm{Gx}_2$ for all $x\in\R^d$.  We condition on this event in the rest of the proof. 

Second, we show that the error between$\Norm{I_d - (SU)^\top (SU)}_2$ and $\Norm{I_d - (\hat SU)^\top(\hat SU)}_2$ is small. 
\begin{align*}
    \Norm{I_d - (\hat SU)^\top(\hat SU)}_2 &\leq \Norm{I_d - (SU)^\top (SU)}_2 + \Norm{(\hat SU)^\top(\hat SU)-(SU)^\top(SU)}_2\\ &\leq \eps + \Norm{(\hat SU)^\top(\hat SU)-(SU)^\top(SU)}_2.
\end{align*}

Observe that $(\hat SU)^\top(\hat SU)=\sum_{i=1}^s\frac{U_{j(i)}^\top U_{j(i)}}{s\hat p_{j(i)}}=\sum_{i=1}^s \frac{U_{j(i)}^\top U_{j(i)}}{(1\pm O(\eps))sp_{j(i)}} = \frac{(SU)^\top(SU)}{1\pm O(\eps)}$ and thus
\[
\Norm{(\hat SU)^\top(\hat SU)-(SU)^\top(SU)}_2 = O(\eps)\Norm{(SU)^\top(SU)}_2.
\]

We have proved that $\Norm{I_d - (SU)^\top (SU)}_2 \leq \eps$, so we have  $\Norm{I_d - (\hat SU)^\top(\hat SU)}_2 \leq \eps + O(\eps)(1+\eps) = O(\eps)$. 
By rescaling $\eps'$, we can assume that $\Norm{I_d - (\hat SU)^\top (\hat SU)}_2 \leq \eps$.

Now consider the subspace spanned by the columns of $M$ together with $b$. For any vector $y=Gx-b$, $\Norm{\hat Sy}_2=(1\pm\eps)\Norm{y}_2$. Recall that we have defined $\hat F_i=\frac{E_i}{\sqrt{s \hat p_i}}$, where $\hat F_i$ and $E_i$ are the corresponding $i$-th row of $F$ and $E$. Let $\hat F^{(1)}$ be the first $d$ columns of $\hat F$ and $\hat F^{(2)}$ be the last column of $\hat F$. Hence, the original linear regression problem can be written as $\min \Norm{(\hat SG+\hat F^{(1)})x-(\hat Sb+\hat F^{(2)})}_2$. 

Note that $\tilde{x} = \arg\min_x \Norm{(\hat SG+\hat F^{(1)})x-(\hat Sb+\hat F^{(2)})}_2$ satisfies 
\begin{align*}
    \min_{\tilde x}\Norm{(\hat SG+\hat F^{(1)})\tilde{x}-(\hat Sb+\hat F^{(2)})}_2 
    &\leq \Norm{(\hat SG+\hat F^{(1)})x^*-(\hat Sb+\hat F^{(2)})}_2\\
    &\leq \Norm{\hat S(Gx^*-b)}_2+\Norm{\hat F^{(1)}x^*-\hat F^{(2)}}_2\\
    &\leq  (1+\eps)\Norm{Gx^*-b}_2+\Norm{\hat F}_2\sqrt{\Norm{x^*}_2^2+1},
\end{align*}

where the third inequality holds because $\hat S$ is a subspace embedding for the column space of $G$ together with $b$ and $x^*=\arg \min_{x \in \mathbb{R}^d} \Norm{Gx-b}_2$.

Now, consider the upper bound on $\Norm{\hat F}_2$. Since 
\[
\Norm{\hat F_i}_2^2 = \frac{\Norm{ E_i}_2^2}{s\hat p_i}\leq \gamma^2 \frac{\Norm{G_i}_2^2 + |b_i|^2}{sc(\Norm{G_i}_2^2+\Abs{b_i}^2)}(\Norm{G}_F^2+\Norm{b}_2^2)\leq \frac{\gamma^2}{sc}(\Norm{G}_F^2+\Norm{b}_2^2)
\]
and
\[
\Norm{x^*}_2 = \Norm{G^\dagger b}_2\leq \frac{\Norm{b}_2}{\sigma_{\min}(G)},
\]
we have that
\begin{align*}
    \min_{\tilde x}\Norm{(\hat SG+\hat F^{(1)})\tilde{x}-(\hat Sb+\hat F^{(2)})}_2 &\leq (1+\eps)\Norm{Gx^*-b}_2 + \Norm{\hat F}_2\sqrt{\Norm{x^*}_2^2+1}\\
    &\leq (1+\eps)\Norm{Gx^*-b}_2 + \frac{\gamma}{\sqrt{c}}\sqrt{\Norm{G}_F^2+\Norm{b}_2^2}\cdot\sqrt{\frac{\Norm{b}_2^2}{\sigma_{\min}^2(G)}+1}\\
    &\leq (1+\eps)\Norm{Gx^*-b}_2 + \frac{\gamma}{\sqrt{c}}\left(\sqrt{\Norm{G}_F^2+\Norm{b}_2^2}+\sqrt{d+\frac{\Norm{b}_2^2}{\Norm{G}_2^2}}\kappa\Norm{b}_2\right).
\end{align*}

By our assumption, $c=1-O(\eps)$ and $\gamma=O(\eps)$. Rescaling $\eps$ gives the claimed bound, completing the proof of Theorem~\ref{thm:linear-regression}.
\end{proof}

\end{document}